\documentclass[a4paper,10pt]{amsart}

\usepackage[mathscr]{eucal}

\numberwithin{equation}{section}
\newtheorem{theorem}{Theorem}
\newtheorem{lemma}{Lemma}
\newtheorem{proposition}{Proposition}
\newtheorem{corollary}{Corollary}

\theoremstyle{definition}

\newtheorem{example}[theorem]{Example}

\theoremstyle{remark}
\newtheorem{remark}{Remark}
\newcommand{\I}{\mathbf{1}}
\newcommand {\Cov}{\mbox{Cov}}

\newcommand{\cadlag}{c\`adl\`ag }

\begin{document}
\title{A Black--Scholes Model with Long Memory}
\author{John A. D. Appleby}
\address{Edgeworth Centre for Financial Mathematics, School of Mathematical
Sciences, Dublin City University, Glasnevin, Dublin 9, Ireland}
\email{john.appleby@dcu.ie}
\urladdr{http://webpages.dcu.ie/\textasciitilde applebyj}

\author{John A. Daniels}
\address{Edgeworth Centre for Financial Mathematics, School of Mathematical
Sciences, Dublin City University, Glasnevin, Dublin 9, Ireland}
\email{john.daniels2@mail.dcu.ie}

\author{Katja Krol}
\address{Humboldt University, Institute of Mathematics, Unter den
Linden 6, 10099 Berlin, Germany} \email{krol@mathematik.hu-berlin.de}

\thanks{John Appleby and John Daniels were partially funded by the Science
Foundation Ireland grant 07/MI/008 ''Edgeworth Centre for Financial
Mathematics''. John Daniels was also partially funded by The Embark
Initiative operated by the Irish Research Council for Science,
Engineering and Technology (IRCSET) under the project ''Volatility
Models in Inefficient Markets''. Katja Krol was supported by the
Deutsche Telekom Stiftung.}


 \begin{abstract}
This note develops a stochastic model of asset volatility. The
volatility obeys a continuous--time autoregressive equation.
Conditions under which the process is asymptotically stationary and
possesses long memory are characterised. Connections with the class
of ARCH$(\infty)$ processes are sketched.
  \end{abstract}


\maketitle

\section{Introduction}
In this paper we consider the autocorrelation of the volatility of
the following stochastic functional differential equation:
\begin{align}\begin{split}\label{eq.Xintro}
dX(t)&=\left(\sigma + \beta\left(\int_{-\tau}^0 X(t+v)\lambda(\, d
v)-\int_0^t X(t-s)\kappa(\,ds)\right)\right)\,dB(t), \quad t\geq 0;\\
X(t)&=0, \quad t\in [-\tau, 0].\end{split} \end{align} Here, the
volatility is the process $V$ such that $dX(t)=V(t)\,dB(t)$. The
form of equation \eqref{eq.Xintro} is in part motivated by models of
volatility in financial mathematics in which some traders use past
information about the market to determine their investment
strategies. As indicated in the next section, this leads us to
assume that $\lambda$ and $\kappa$ are finite measures without singular
parts such that
\begin{equation} \label{eq.equalweight}
\lambda[-\tau, 0]=\int_{-\tau}^0  \lambda(ds)=\int_0^\infty \kappa(ds)=\kappa[0, \infty).
\end{equation}
In the context of this work, we show that $X$ can be thought of as
de--trended market returns, and hence \eqref{eq.Xintro} leads
automatically to a Black--Scholes type model with memory.

In financial markets, it is of practical interest to determine
whether market returns or other important indicators, such as the
volatility $V$, possess predictable components. Therefore, under the
condition \eqref{eq.equalweight}, we give necessary and sufficient
conditions under which $V$ is an asymptotically weakly stationary
process, with non--trivial limiting autocovariance function. We do
this by establishing that $V$ is a continuous--time analogue of
solutions of stochastic difference equations structurally related to
the class of ARCH($\infty$) processes.

It is also of interest to see whether such processes in finance possess
\emph{long memory} or \emph{long range dependence} (\cite{Cont,
taqqu}), in the usual sense that the limiting autocovariance
function $\gamma$ of $V$ has the property
\begin{equation} \label{eq.longmem}
\int_{0}^\infty \gamma(s)\,ds = +\infty.
\end{equation}
We are able to characterise whether $V$ has long memory or not in
the case when $\kappa$ is a positive measure, by proving that
\eqref{eq.longmem} holds if and only if $\kappa$ has infinite first
moment. We also establish the exact rate of convergence to zero of
$\gamma(t)$ as $t\to\infty$ in the case when $\kappa$ is absolutely
continuous with regularly varying density $k$.

We consider also discrete analogues of \eqref{eq.Xintro}, and demonstrate that the stationarity
of the volatility and presence of long range dependence can be characterised
in a similar manner to the continuous case.
\section{Motivation from finance}
Let $S=\{S(t):t\geq 0\}$ be the stock price of a single risky asset
whose evolution is governed by
\begin{equation} \label{eq.stock}
dS(t)=\mu S(t)\, dt+S(t)\,dX(t), \quad t\geq 0; \quad S(0)=s_0>0,
\end{equation}
and $X$ obeys \eqref{eq.Xintro}. This means that shares in the stock
started trading at time $t=0$. In what follows, we assume that $\lambda$
and $\kappa$ have finite total variation, which implies that there is a
unique continuous adapted processes $X$ which satisfies
\eqref{eq.Xintro} and which is moreover a semimartingale. Therefore
there is a unique positive continuous adapted processes $S$ which satisfies
\eqref{eq.stock}. It is reasonable to
call the process $V=\{V(t):t\geq 0\}$ which is defined by
\begin{equation} \label{def.vol}
X(t)=\int_0^t V(s)\,dB(s), \quad t\geq 0,
\end{equation}
the \emph{volatility} of the stock price because from
\eqref{eq.stock} and \eqref{eq.Xintro} we have
\[
dS(t)=\mu S(t)\,dt + V(t)S(t)\,dB(t), \quad t\geq 0.
\]

We now motivate the form of \eqref{eq.Xintro}, and in particular
begin by explaining the economic interpretation of $X$. The
\emph{cumulative return} $R=\{R(t):t\geq 0\}$ on the stock is
defined by the identity $dS(t)=S(t)\,dR(t)$ for $t\geq 0$ and
$R(0)=0$. From this, \eqref{eq.stock} and \eqref{def.vol} we see
that $X$ is the \emph{de--trended cumulative return}, because
$X(t)=R(t)-\mu t$ for $t\geq 0$.

At time $t\geq 0$, traders in the market take a weighted average of
the de--trended returns over the last $\tau>0$ units of time, giving
a short--run indicator of returns $\int_{-\tau}^0 X(t+s)\lambda(ds)$.
They also form a long--run indicator of returns by taking a weighted
average of the de--trended returns over the entire history of the
asset, according to $\int_0^\infty X(t-s)\,\kappa(ds)$. Since there
is no trading before time $t=0$, we set $X(t)=0$ for $t\leq 0$, so
that the long--term indicator is also given by $\int_0^t
X(t-s)\,\kappa(ds)$. In order that the indicators represent weighted
averages with the same weight, we require that $\lambda$ and $\kappa$ obey
\eqref{eq.equalweight}. The traders believe that these indicators
signal that the market is far from equilibrium whenever the
indicators differ significantly, and this causes the traders to
trade greater amounts of the stock. It has been observed in real
financial markets that the volume of trade is positively correlated
with the volatility of the asset (see e.g.~\cite{epps},
\cite{tauchen} and the references therein), which leads to the
simple model that the volatility depends linearly on the trading
volume, which itself depends on the difference between these
indicators. By this reasoning, we arrive at
\begin{equation}  \label{eq.VX}
V(t)=\sigma+\beta\left(\int_{-\tau}^0 X(t+v)\lambda(\, d v)-\int_0^t
X(t-s)\kappa(\,ds)\right), \quad t\geq 0.
\end{equation}
Therefore, using \eqref{eq.VX} and \eqref{def.vol}, we see that $X$
obeys \eqref{eq.Xintro}, because $X$ is identically zero on
$(-\infty,0]$. It can be readily shown for $\Delta>\delta\ge 0$ that
\[
\Cov(R(t+\delta)-R(t),R(t+\delta+\Delta)-R(t+\Delta))=0, \quad t\geq
0,
\]
so the $\delta$--returns over non--overlapping time intervals are uncorrelated
at all time horizons. Hence the market is efficient in the sense of Fama,
see e.g.\cite{Fama:70}.

The parameter $\beta$ represents the sensitivity of these
trend--following traders; if it is large, the traders are sensitive
and have a large impact on the price dynamics. If no such traders
were present, then $\beta=0$, and $S$ obeys the classical
Black--Scholes stochastic differential equation.

In the case when $\beta\neq 0$, our model of asset price evolution
depends on the path of the price process, and the returns follow a
stochastic functional differential equation (SFDE). Other models of
financial markets where price evolution is described by a SFDE
include \cite{anh, Ap_Rie_Swo, App_Dan, bouchaud, hobson}.

\section{Mathematical preliminaries}
By $M(I)$  we denote the set of all signed  $\sigma$-finite Borel
measures on $I\subseteq \mathbb{R}$ with values in $\mathbb{R}$.
Let $|\kappa|$ and  $\|\kappa\|$ denote the variation and the total
variation of a measure $\kappa\in M(I)$ respectively.

Let $(\Omega, \mathcal{F}, \mathbb{P})$ be a complete probability
space equipped with a filtration $\mathbb{F}=(\mathcal{F}_{t})_{t
\geq 0}$, and let $B=\{B(t): t\geq 0\}$ be a one--dimensional
Brownian motion on this probability space. Let $\mathbb{D}$ denote
the set of all adapted \cadlag processes. By $\mathcal{H}^*$ we
denote the set of all $\mathbb{F}$--adapted processes $X=\{X(t):t\ge
0\}$, satisfying
\begin{equation} \label{eq.Xfinitemsq_general}
\mathbb{E}\left[\max_{0\leq t\leq T} |X(t)|^2 \right]<+\infty, \quad
\text{for every $T>0$.}
\end{equation}
  Let $\sigma,\beta\in \mathbb{R}$. Suppose that the measure $\kappa$ has finite total variation. Then, the following stochastic differential equation has a unique strong solution.
\begin{equation}\label{eq.X_general}
dX(t)=\left(\sigma + \beta\left(\int_0^t
X(t-s)\kappa(\,ds)\right)\right)\,dB(t), \quad t\geq 0;\qquad  X(0)=0.
\end{equation} To see this,  we introduce the well--defined operator
$F\,:\, \mathbb{D}\to \mathbb{D}$
\begin{equation} \label{def.V_general}
F(X\I_{[0, t]})=\beta\int_0^t X(t-s)\kappa(\,ds), \quad t\geq 0.
\end{equation}
The process $F$ is functional Lipschitz with $F(0)=0$ in the sense
of   \cite[p.~250]{protter}, since it satisfies for two \cadlag
processes $X, Y$
\[ |F(X\I_{[0, t]})-F(Y\I_{[0, t]})|\le |\beta|\|\kappa\|\sup_{s\le t} |X(s)-Y(s)|,\]
almost surely for each $t\ge 0$. Hence, the equation
\eqref{eq.X_general} fulfills all the assumptions of Lemma V.2,  in
\cite{protter} and  has a unique strong solution.  We refer to this
process $X$ as the solution of \eqref{eq.X_general}. Moreover, by
Doob's inequality (cf.~\cite[Theorem I.20]{protter}) the solution
$X$ belongs to $\mathcal{H}^*$.

We denote the spaces of real--valued  integrable and continuous
functions by $L^1(0, \infty)$ and
$C([0,\infty);\mathbb{R})$ respectively. Then, $L^1_{\text{loc}}(0,
\infty)$ denotes the space of all Lebesgue measurable
functions, whose restrictions to compact subsets of $\mathbb{R}$
belong to $L^1$.

We write $f\sim g$ for $x\to x_0\in \mathbb{ R}\cup \{\pm \infty\}$
if $ \lim_{x\to x_0} f(x)/g(x)=1.$
 A function $L\,:\,[0, \infty)\to (0, \infty)$ is {\em slowly varying at infinity} if  $\lim_{t\to \infty} L(xt)/L(t)=1$
 holds for  all $x>0.$
A function $f$ {\em varies regularly with index $\alpha\in
\mathbb{R}$}, $f\in \mbox{RV}_\infty(\alpha)$,  if it is of the form
$
 f(t)=t^\alpha L(t)$
 with $L$ slowly varying, see e.g.~\cite[Ch.\ VIII.8]{Feller:1}.
\label{prelim}
\subsection{Assumptions on Equation \eqref{eq.Xintro}}
Next, we give some concrete assumptions under which
\eqref{eq.Xintro} has a well--defined and unique solution, and
introduce some useful notation. Let $\sigma,\beta\in \mathbb{R}$,
$\tau\in \mathbb R_+$. Consider the stochastic differential equation
with delay given by \eqref{eq.Xintro}. We assume that the two
measures $\kappa\in M(\mathbb R_+)$ and $\lambda\in M([-\tau, 0])$
are decomposable into  absolutely continuous and discrete parts:
there exist Lebesgue integrable functions $k\in L^1[0, \infty)$,
$\ell\in L^1[-\tau, 0]$, real--valued sequences $(\kappa_j)_{j\in
\mathbb{N}}$, $(\lambda_j)_{0\le j\le N} $, $N>0$, and monotone
increasing  positive sequences  $(\rho_j)_{j\in \mathbb{N}}$,
$(\tau_j)_{0\le j\le N}$, $\tau_N\le\tau$, so that
\begin{equation}
\kappa(\,ds)=\sum_{j=0}^\infty \kappa_j \I_{\{\rho_j\}}(\,ds)+k(s)\,ds,
 \quad \lambda(\,ds)=\sum_{j=0}^N \lambda_j \I_{\{-\tau_j\}}(\,ds) +\ell(s)\,ds.\end{equation}
Moreover, $\kappa$ and $\lambda$ satisfy
\begin{equation} \label{eq.khypot}
k\in L^1[0,\infty), \quad  \sum_{j=0}^\infty |\kappa_j|<\infty;\qquad
\int_0^\infty \kappa(\,ds)=\int_{-\tau}^0 \lambda(\,ds).
\end{equation}
The last equality in equation \eqref{eq.khypot} can be written as
\begin{equation} \label{eq.kweight} \int_0^\infty k(s)\,ds
+\sum_{j=0}^\infty \kappa_j =\int_{-\tau}^0 \ell(s)\,
ds+\sum_{j=0}^N \lambda_j.\end{equation} Then, as stated in Section
\ref{prelim}, \eqref{eq.Xintro} has a unique strong solution.
Moreover, this solution obeys $X\in \mathcal{H}^*$. If $\beta=0$,
$X$ is standard Brownian motion. If $\beta\neq 0$, but $\sigma=0$,
it transpires that $X(t)=0$ for all $t\geq 0$ almost surely. For
this reason, we suppose that $\sigma\neq 0$ and $\beta\neq 0$.

By hypothesis we see that $  K :[0,\infty)\to
\mathbb{R}$ given by
\begin{align} \label{def.kappa}
 K (x)
&=-\int_{-\tau}^{-(x\wedge \tau)} \lambda(\, ds) + \int_x^\infty
\kappa(\,ds), \quad x\ge 0,\end{align} is well--defined. Moreover,
$ K $ can be written as\begin{equation} \label{kappa_concrete}
 K (x)=\left\{ \begin{array}{ll}\sum_{\{j\,:\,\rho_j\ge x\}}\kappa_j +  \int_x^\infty k(s)\,ds, \quad &x\geq \tau,\\
-\sum_{\{ j\,:\, \tau_j\ge x\}}\lambda_j-\int_{-\tau}^{-x} \ell(s)\, d
s+\sum_{\{j\,:\,\rho_j\ge x\}}\kappa_j+\int_x^\infty k(s)\, ds, \quad
&x<\tau\end{array}\right. .
\end{equation}

\section{Autoregression of the Volatility Process}
We introduce the well--defined process $V=\{V(t):t\geq 0\}$
\begin{equation} \label{def.V}
V(t)=\sigma + \beta\left(\int_{-\tau}^0 X(t+v)\lambda(\, d v) -\int_0^t
X(t-s)\kappa(\,ds)\right), \quad t\geq 0.
\end{equation}
Since $X\in \mathcal{H}^*$ and by \eqref{eq.khypot} we have that
$V\in \mathcal{H}^*$ and $ \mathbb{E}[V(t)]=\sigma, \; t\geq 0. $ By
\eqref{eq.Xintro} and \eqref{def.V} we have $dX(t)=V(t)\,dB(t)$.
Therefore $V$ is the volatility process. We see also that
\begin{equation*}
\mathbb{E}[X(t)]=0, \quad \mathbb{E}[X^2(t)]=\int_0^t
\mathbb{E}[V^2(s)]\,ds, \quad t\geq 0.
\end{equation*}

We show that $V$ obeys a linear stochastic integral equation, and
deduce that $t\mapsto\mathbb{E}[V^2(t)]$ satisfies a linear Volterra
integral equation.
\begin{proposition}  \label{proposition.Vint}
Suppose that $\kappa$ obeys \eqref{eq.khypot}. Then $V$ defined by
\eqref{def.V} obeys
\begin{equation} \label{eq.Vinteqn}
V(t)=\sigma+\beta \int_0^t  K (t-s)V(s)\,dB(s), \quad t\geq 0.
\end{equation}
Moreover,
\begin{equation} \label{eq.msqVinteqn}
\mathbb{E}[V^2(t)]=\sigma^2+ \beta^2 \int_0^t
 K ^2(t-s)\mathbb{E}[V^2(s)]\,ds, \quad t\geq 0.
\end{equation}
\end{proposition}
\begin{proof}
Since $V\in \mathcal{H}^*$, by a stochastic Fubini's theorem  (e.g.
\cite[Ch.IV.6, Thm.~65]{protter}) we have
\begin{align*}
V(t)&=\sigma + \beta\left(\int_{-\tau}^0 X(t+s)\lambda(\, ds) -\int_0^tX(t-s) \kappa(\,ds)\right)\\
&=\sigma + \beta\left(\int_0^t \left\{\int_{-\tau \vee
(u-t)}^0\lambda(\, ds) -\int_0^{t-u} \kappa(\,ds)\right\}
V(u)\,dB(u)\right).
\end{align*}
Now by \eqref{def.kappa} and \eqref{eq.kweight} we have
\begin{align*}
V(t)
&=\sigma + \beta\left(\int_0^t \left\{-\int_{-\tau}^{-\tau \vee (u-t)} \lambda(\, ds)
+\int_{t-u}^\infty \kappa(\,ds)\right\} V(u)\,dB(u)\right)\\
&=\sigma + \beta\int_0^t  K (t-u) V(u) \,dB(u),
\end{align*}
as required. By \eqref{eq.Vinteqn} we have for each $t\geq 0$
\begin{equation*}
V^2(t)=\sigma^2+ 2\sigma\beta\int_0^t  K (t-s)V(s)\,dB(s) +
\beta^2 \left(\int_0^t  K (t-s)V(s)\,dB(s)\right)^2.
\end{equation*}
By considering $t\geq 0$ as fixed, and using the fact that $V\in
\mathcal{H}^*$, we can apply It\^o's isometry to get
\eqref{eq.msqVinteqn} as required.
\end{proof}
Given that $V$ has constant expectation, it is interesting to ask
whether its variance (or equivalently, its second moment) settles
down. From \eqref{eq.msqVinteqn} we can readily determine necessary
and sufficient conditions for it to do so.
\begin{proposition} \label{prop.varVconst}
Suppose that $\kappa$ obeys \eqref{eq.khypot} and that $ K $  obeys
$ K \in L^2(0,\infty)$. Suppose that $V$ is defined by
\eqref{def.V}. If
\begin{equation} \label{eq.kappaintsmall}
\beta^2 \int_0^\infty  K ^2(s)\,ds < 1,
\end{equation}
then
\begin{equation}    \label{eq.msqlimit}
\lim_{t\to\infty}
\mathbb{E}[V(t)^2]=\frac{\sigma^2}{1-\beta^2\int_0^\infty
 K ^2(s)\,ds}.
\end{equation}
\end{proposition}
The fact that $\lim_{t\to\infty} \mathbb{E}[V^2(t)]$ is always greater than
$\sigma^2$ shows that the presence of the trend following speculators increases
market volatility relative to the level $\sigma^2$, which would be obtained in
their absence (where $\beta=0$). In other words, the presence of these traders makes the market more risky,
and leads to greater fluctuations. This is similar to findings of \cite{delong}, in which the presence of noise traders
increases the risk for informed investors.
\begin{proof}
From representation \eqref{kappa_concrete} we see that satisfies $ K
\in L^1_{\text{loc}}(0, \infty)$. Therefore, Theorem 2.3.1 in
\cite{GrLoSt90} applies and there exists a unique solution $r\in
L^1_{\text{loc}}(0, \infty)$ of
\[
r(t)=\beta^2  K ^2(t)+\beta^2\int_0^t  K ^2(t-s)r(s)\,ds,
\quad t\geq 0.
\]
Since $ K ^2$ is  nonnegative, the iteration method in the proof
of theorem 2.3.1 in \cite{GrLoSt90} yields, that the resolvent $r$
is also nonnegative. By Theorem 2.3.5 in in the same book, the
process $\mathbb{E}[V^2(t)]$ is continuous and satisfies
\[
\mathbb{E}[V^2(t)]=\sigma^2+\int_0^t
r(t-s)\sigma^2\,ds=\sigma^2\left(1+\int_0^t r(s)\,ds\right).
\]
Therefore if $\int_0^\infty r(s)\,ds<+\infty$, we have the desired
result. Define $a:[0,\infty)\to[0, \infty)$ by
$a(t):=\beta^2\int_0^t  K ^2(s)\,ds$ for $t\geq 0$.  For any
$T>0$ we have
\begin{align*}
\int_0^T r(t)\,dt &=\beta^2 \int_0^T  K ^2(t)\,dt +\beta^2\int_0^T \int_0^t  K ^2(t-s)r(s)\,ds\,dt\\
&=\beta^2 \int_0^T  K ^2(t)\,dt +\beta^2\int_0^T \int_s^T  K ^2(t-s)\,d t \,r(s)\,ds   \\
&=a(T) +\int_0^T a(T-s)r(s)\,ds.
\end{align*}
Since $a(t)\uparrow \alpha:=\beta^2\int_0^\infty  K ^2(s)\,ds<1$
as $t\to\infty$, for all $T\geq 0$ we have $ \int_0^T r(t)\,dt \leq
\alpha +\alpha\int_0^T  r(s)\,ds.$ Since $\alpha\in (0,1)$ we have $
\int_0^T r(t)\,dt \leq \alpha/(1-\alpha)$, for all $T\geq 0$.
Therefore we have $r\in L^1(0,\infty)$. Since $a$ converges as $T\to
\infty$ and $r\in L^1(0,\infty)$ we have that
\[
\lim_{T\to\infty} \int_0^T a(T-s)r(s)\,ds=\alpha \int_0^\infty
r(s)\,ds.
\]
Therefore we have $ \int_0^\infty r(s)\,ds = \alpha +
\alpha\int_0^\infty r(s)\,ds, $ from which we infer
\[
\int_0^\infty r(s)\,ds = \frac{\alpha}{1-\alpha}
=\frac{\beta^2\int_0^\infty  K ^2(s)\,ds}{1-\beta^2\int_0^\infty
 K ^2(s)\,ds}.
\]
Therefore we have
\[
\lim_{t\to\infty}\mathbb{E}[V^2(t)]=\sigma^2\left(1+\int_0^\infty
r(s)\,ds\right) =\sigma^2\left(1+\frac{\beta^2\int_0^\infty
 K ^2(s)\,ds}{1-\beta^2\int_0^\infty  K ^2(s)\,ds} \right)
\]
which confirms the result.
\end{proof}

\begin{remark}  \label{rem.msqtoinfty}
In the case 
 $ K \not \in L^2(0,\infty)$,  the
solution of \eqref{eq.msqVinteqn} obeys $ \lim_{t\to\infty}
\mathbb{E}[V(t)^2]=+\infty. $ To see this notice first, by
\eqref{eq.msqVinteqn}, that $\mathbb{E}[V^2(t)]\geq \sigma^2$ for
all $t\geq 0$. Therefore by \eqref{eq.msqVinteqn}, we have
\[
\mathbb{E}[V^2(t)]=\sigma^2+ \beta^2 \int_0^t
 K ^2(t-s)\mathbb{E}[V^2(s)]\,ds
\ge \sigma^2+ \sigma^2\beta^2 \int_0^t  K ^2(s)\,ds.
\]
Therefore as $t\to\infty$ and $ K \not \in L^2(0,\infty)$, we
have that $\mathbb{E}[V^2(t)]\to\infty$ as $t\to\infty$.
\end{remark}

\section{Asymptotic stationarity of $V$ and Long Memory in $V$}
In our next result, we show that  $ K \in L^2(0,\infty)$ and
$ K $ obeying \eqref{eq.kappaintsmall} are necessary conditions
for $V$ to be asymptotically stationary. To fix terminology, we say
that a real scalar process $U=\{U(t):t\geq 0\}$ is \emph{(weakly)
asymptotically stationary} if there exists $\theta\in \mathbb{R}$ and a
function $\gamma:[0,\infty)\to \mathbb{R}$ such that
$\lim_{t\to\infty}\mathbb{E}[U(t)]=\theta$ and
$\lim_{t\to\infty}\Cov(U(t),U(t+\Delta))=\gamma(\Delta)$ for each
$\Delta\geq 0$.
\begin{lemma} \label{lemma.necctns}
Suppose that $\kappa$ obeys \eqref{eq.khypot}. Suppose that $V$ defined
by \eqref{def.V} is asymptotically stationary. Then $ K \in
L^2(0,\infty)$ and $ K $ obeys \eqref{eq.kappaintsmall}.
\end{lemma}
\begin{proof}
Since $V$ is asymptotically stationary, it follows that there is a
finite $g\geq 0$ such that $ g:=\lim_{t\to\infty}
\text{Cov}(V(t),V(t)). $ Since $\mathbb{E}[V(t)]=\sigma$ for all
$t\geq 0$, we have that there is an $a\in \mathbb{R}$ such that
\begin{equation} \label{eq.msqVttolimit}
a:=\lim_{t\to\infty} \mathbb{E}[V(t)^2]=\lim_{t\to\infty}
\left\{\text{Cov}(V(t),V(t))+\mathbb{E}[V(t)]^2\right\} =g+\sigma^2.
\end{equation}
If $ K \not\in L^2(0,\infty)$, we have by
Remark~\ref{rem.msqtoinfty} that $\lim_{t\to\infty}
\mathbb{E}[V(t)^2]=+\infty$, which contradicts
\eqref{eq.msqVttolimit}. Therefore we must have that $ K \in
L^2(0,\infty)$.

Since $V$ obeys \eqref{eq.msqVttolimit}, we see that $a\geq
\sigma^2>0$. Since $ K \in L^2(0,\infty)$ we have
\[
\lim_{t\to\infty} \int_0^t  K ^2(t-s)\mathbb{E}[V^2(s)]\,ds =
\int_0^\infty  K ^2(s)\,ds \cdot a.
\]
Therefore from  \eqref{eq.msqVinteqn} we have that
\[
a=\sigma^2+ \beta^2 \int_0^\infty
 K ^2(s)\,ds \cdot a.
\]
Since $a>0$, we must have $\beta^2\int_0^\infty  K ^2(s)\,ds< 1$, as
required.
\end{proof}
\begin{remark} \label{rem.msqconverge}
Perusal of the proof of Lemma~\ref{lemma.necctns} shows that if
$\kappa$ obeys \eqref{eq.khypot}, $ K $ is defined by
\eqref{def.kappa}, and
\begin{equation} \label{eq.limit_finite}
\lim_{t\to\infty} \mathbb{E}[V^2(t)] \quad\text{ exists and is
finite},
\end{equation}
then $ K \in L^2(0,\infty)$ and \eqref{eq.kappaintsmall} holds.
Therefore, by this remark and Proposition~\ref{prop.varVconst}, we
see that $ K $ obeying $ K \in L^2(0,\infty)$ and
\eqref{eq.kappaintsmall} is equivalent to \eqref{eq.limit_finite},
and that both imply that the limit is equal to
$\sigma^2/(1-\beta^2\int_0^\infty  K ^2(s)\,ds)$.
\end{remark}

In our next result, we show that the conditions imposed on $ K $
in Proposition~\ref{prop.varVconst}
 are necessary and sufficient for $V$ to be asymptotically stationary. Moreover, we determine a formula
 for the limiting autocovariance function of $V$.
\begin{theorem} \label{thm.Vtns}
Suppose that $\kappa$ obeys \eqref{eq.khypot}. Then the following
statements are equivalent.
\begin{itemize}
\item[(A)] $ K \in L^2(0,\infty)$ and $ K $ obeys  \eqref{eq.kappaintsmall};
\item[(B)] The process $V$ defined by \eqref{def.V} is asymptotically stationary.
\end{itemize}
Moreover, both imply that the function
$\gamma:[0,\infty)\to\mathbb{R}$ given by
\begin{equation} \label{def.gamma}
\gamma(\Delta)= \beta^2  \frac{\sigma^2}{1-\beta^2\int_0^\infty
 K ^2(s)\,ds}\cdot \int_0^\infty  K (s) K (s+\Delta)\,ds,
\quad \Delta\geq 0,
\end{equation}
is well--defined, and that $ \mathbb{E}[V(t)]=\sigma$, for all
$t\geq 0$,
\begin{equation}\lim_{t\to\infty} \text{Cov}(V(t),V(t+\Delta))
=\gamma(\Delta), \quad \text{for all $\Delta\geq
0$}.\label{eq.cov_limit}\end{equation}
\end{theorem}
\begin{proof}
In Lemma~\ref{lemma.necctns}, we have shown that statement (B)
implies statement (A). Suppose statement (A) holds. Let $t\geq 0$.
Since $V$ obeys \eqref{eq.Vinteqn}, we have
\[
V(t) =\mathbb{E}[V(t)]+ \beta\int_0^t  K (t-s)V(s)\,dB(s).
\]
Therefore it follows with $V\in \mathcal{H}^*$ and $ K \in
L^2(0,\infty)$ that
\begin{align*}
\text{Cov}(V(t),V(t&+\Delta)) \\
 &=\mathbb{E}\left[\beta\int_0^t
 K (t-s)V(s)\,dB(s)\cdot \beta\int_0^{t+\Delta}
 K (t+\Delta-s)V(s)\,dB(s)  \right].
 \end{align*}
  Since $\Delta\geq 0$, for
each $t\geq 0$ fixed we have
\[
\text{Cov}(V(t),V(t+\Delta)) =\beta^2 \int_0^t
 K (t-s) K (t+\Delta-s) \mathbb{E}[V^2(s)]\,ds.
\]
For $\tau\geq 0$ define
$ K _\Delta(\tau)= K (\tau) K (\tau+\Delta)$. Then
\begin{equation} \label{eq.covVttdel}
\text{Cov}(V(t),V(t+\Delta)) =\beta^2 \int_0^t
 K _\Delta(t-s)\mathbb{E}[V^2(s)]\,ds.
\end{equation}
Since $2|xy|\leq x^2+y^2$ for all $x,y\in\mathbb{R}$, we have $
0\leq | K _\Delta(\tau)|\leq
1/2 K ^2(\tau)+1/2 K ^2(\tau+\Delta). $ Since $ K \in
L^2(0,\infty)$, it follows that $ K _\Delta\in
L^1([0,\infty);[0,\infty))$. By Proposition~\ref{prop.varVconst}, we
have that $t\mapsto \mathbb{E}[V^2(t)]$ obeys \eqref{eq.msqlimit}.

Therefore it follows that
\[
\lim_{t\to\infty}\int_0^t  K _\Delta(t-s)\mathbb{E}[V^2(s)]\,ds =
\int_0^\infty  K _\Delta(s)\,ds \cdot \lim_{t\to\infty}
\mathbb{E}[V^2(t)],
\]
which, by \eqref{eq.covVttdel} and \eqref{def.gamma} implies
\eqref{eq.cov_limit}.

Therefore we have that there is a function $\gamma$, defined by
\eqref{def.gamma}, such that \eqref{eq.cov_limit} holds true.
Furthermore we have that $\mathbb{E}[V(t)]=\sigma$. Thus $V$ is
asymptotically stationary, which proves (B). Hence (A) and (B) are
equivalent. Moreover, we have shown that $\mathbb{E}[V(t)]$ and
$\text{Cov}(V(t),V(t+\Delta))$ have the desired properties.
\end{proof}

In our next result we show that $V$ has short memory or long memory
 according as to whether $ K $ is integrable or not.
\begin{theorem} \label{thm.longmemshortmem}
Suppose that $\kappa$ obeys \eqref{eq.khypot} and that $ K $
satisfies $ K \in L^2(0,\infty)$ and obeys
\eqref{eq.kappaintsmall}.
\begin{itemize}
\item[(a)] If $\kappa $ obeys
$\int_0^\infty s|\kappa|(\, ds)<\infty$, then $\gamma$ defined by
\eqref{def.gamma} obeys $ \int_0^\infty
|\gamma(\Delta)|\,d\Delta<+\infty.$
\item[(b)]
If $\kappa$ obeys $\int_0^\infty s|\kappa|(\, ds) =+\infty$, and $\kappa$ is
a non--negative measure, then $\gamma$ defined by \eqref{def.gamma}
obeys $
 \int_0^\infty |\gamma(\Delta)|\,d\Delta=+\infty.
$\end{itemize}
\end{theorem}
\begin{remark}
In the case that $\kappa$ obeys $\int_0^\infty s|\kappa|(\, ds)<\infty$ it
follows that $ K \in L^1(0,\infty)\cap C([0,\infty);\mathbb{R})$
and that $ K (t)\to 0$ as $t\to\infty$. Therefore we have
automatically that $ K \in L^2(0,\infty)$.
\end{remark}
\begin{remark}
Part (b) of the theorem still holds in the case when $\kappa$ is a
non--positive measure, by an almost identical argument. One
implication of this fact is that the sensitivity parameter $\beta$
can be negative in \eqref{eq.Xintro}. In terms of modelling,
therefore, it is the \emph{magnitude} of the difference between the
short and long run indicators that matters, rather than the
difference itself.
\end{remark}
\begin{proof}
If
$\int_0^\infty s|\kappa|(\, ds)<\infty$, then $ K \in
L^1(0,\infty)$:
\begin{align}\begin{split}
\int_0^\infty | K (x)|\, dx
&\le \int_0^\infty s|\kappa|(\, ds)+\int_{-\tau}^0 (-s)|\lambda|(\, ds)\\
&\le  \int_0^\infty s|\kappa|(\, ds) +\tau \|\lambda\|.
\end{split}\label{eq.kappa_integrable}
\end{align}
Set $c:=\beta^2  \sigma^2/(1-\beta^2\int_0^\infty  K ^2(s)\,ds)$.
By the definition of $\gamma$, we have
\begin{align*}
\int_0^\infty |\gamma(\Delta)|\,d\Delta &=c\int_0^\infty
\left|\int_0^\infty  K (s) K (s+\Delta)\,ds\right|\,d\Delta \\
&=
c\int_0^\infty | K (s)| \int_s^\infty | K (u)|\,du\,ds.
\end{align*}
Since $ K \in L^1(0,\infty)$, it follows that
\[
\int_0^\infty |\gamma(\Delta)|\,d\Delta\leq c\int_0^\infty
| K (s)| \int_0^\infty | K (u)|\,du\,ds= c\left(\int_0^\infty
| K (s)|\,ds\right)^2,
\]
so $\gamma$ is in $L^1(0,\infty)$.

To prove the second part, since $\kappa$ is non--negative, it follows
that $ K (t)$ is non--increasing and non--negative for $t\ge
\tau$. Let $f$ be defined by $f(\Delta):=\int_{\lceil \tau
\rceil}^\infty  K (s) K (s+\Delta)\,ds$, where $\lceil
x\rceil$ denotes the smallest integer not less than $x$. Then $f$ is
non--negative and non--increasing. Hence $f$ (and therefore
$\gamma$) is not integrable if and only if
\[
\sum_{n=0}^\infty f(n) = \sum_{n=0}^\infty \int_{\lceil \tau
\rceil}^\infty  K (s) K (s+n)\,ds
\]
is infinite. Since $ K $ is non--negative for $t\ge \tau$, we
have
\begin{equation}
\sum_{n=0}^\infty \int_{\lceil \tau \rceil}^\infty
 K (s) K (s+n)\,ds =\sum_{n=0}^\infty \sum_{j=\lceil \tau
\rceil}^\infty \int_{j}^{j+1}  K (s) K (s+n)\,ds.
\end{equation}
Following the steps  of \eqref{eq.kappa_integrable}, we see that
$ K $ is  not integrable. Since  it is non--negative and
non--increasing for $t\ge \tau$, we have that $\sum_{j=\lceil \tau
\rceil}^\infty  K (j)=\infty$. Therefore
\begin{align*}
\sum_{n=0}^\infty \int_{\lceil \tau \rceil}^\infty
 K (s) K (s+n)\,ds
&=\sum_{n=0}^\infty \sum_{j=\lceil \tau \rceil}^\infty \int_{j}^{j+1}  K (s) K (s+n)\,ds\\
&\geq \sum_{n=0}^\infty \sum_{j=\lceil \tau \rceil}^\infty   K (j+1) K (j+n+1)\\
&= \sum_{j=\lceil \tau \rceil+1}^\infty  K (j) \sum_{l=j}^\infty
 K (l) = \sum_{l=\lceil \tau \rceil+1}^\infty  K (l)
\sum_{j=\lceil \tau \rceil+1}^l  K (j).
\end{align*}
Since $\sum_{j=\lceil \tau \rceil}^\infty  K (j)=\infty$ there is
an $N>0$ such that $\sum_{j=\lceil \tau \rceil+1}^l  K (j)\geq 1$
for all $l\geq N$. Therefore as $ K $ is non--negative for $t\ge
\tau$, we have
\[
\sum_{n=0}^\infty \int_{\lceil \tau \rceil}^\infty
 K (s) K (s+n)\,ds \geq \sum_{l=\lceil \tau \rceil+1}^\infty
 K (l)  \sum_{j=\lceil \tau \rceil+1}^l  K (j) \geq
\sum_{l=N}^\infty  K (l)=+\infty.
\]
Hence $ \sum_{n=0}^\infty f(n)=+\infty, $ as required.
\end{proof}
\section{Exact convergence rates for regularly varying weight--functions}
In the previous section, we gave conditions under which $\gamma$ is
either integrable or non--intergrable, but did not establish the
pointwise rate of decay of $\gamma(\Delta)$ as $\Delta\to\infty$.

In this section, we address this question. First, consider measures
$\kappa$ and $\lambda$ satisfying
\[ \kappa(\, ds)=k(s)\,ds, \quad \lambda(\, ds)=\lambda_0\I_{\{0\}} (\, ds), \]
where $k$ is a continuous integrable kernel, $k\in L^1(0,
\infty)\cap C((0, \infty); (0, \infty))$, and $\lambda_0=\int_0^\infty
k(s)\, ds.$ In this case equation \eqref{eq.Xintro} reads
\begin{align} \label{eq.X_simple}
\begin{split}
dX(t)&=\left(\sigma + \beta\left(\lambda_0 X(t) -\int_0^t
X(t-s)k(s)\,ds\right)\right)\,dB(t), \quad t\geq 0;\\  X(t)&=0,
\quad t\in [-\tau, 0].
\end{split}
\end{align}
We consider kernels of the form $k(t)\sim L(t)t^{-\alpha-1}$ as
$t\to\infty$, for certain positive $\alpha$, where $L$ is a slowly
varying function.
\begin{corollary} \label{th.concrete}
Suppose that $k\in \mbox{RV}_\infty(-1-\alpha)$ with $1/2<\alpha<1$
and $ K $ satisfies \eqref{eq.kappaintsmall}. Then, the function
$\gamma$ defined by \eqref{def.gamma} is not integrable and
satisfies
\begin{equation} \label{eq.exact_rate}
\gamma(\Delta)\sim\beta^2  \frac{\sigma^2}{1-\beta^2\int_0^\infty
 K ^2(s)\,ds}\cdot
\frac{\Gamma(2\alpha-1)\Gamma(1-\alpha)}{\Gamma(\alpha+1)\alpha}\Delta^{1-2\alpha}L(\Delta)^2, \quad \Delta\to \infty.\end{equation}
\end{corollary}
Since $-1<1-2\alpha<0$, the function $\gamma$ is obviously
non--integrable.
\begin{proof}
By Karamata's Theorem (see e.g.~\cite[Theorem 1.5.11]{Goldie}), the
function $ K $, defined as in \eqref{def.kappa} satisfies
\[  K (x)=\int_x^\infty k(s)\, ds\sim \frac{1}{\alpha} L(x)x^{-\alpha}, \quad x\to \infty.\]
Since $k$ is non--negative, $ K $ is also non--negative and
non--increasing. Hence, following the steps of the proof of Theorem
7.1 in \cite{Ap_Kr} we obtain
\begin{equation}\label{eq.Ap_Kr}\ \lim_{\Delta\to \infty} \frac{\gamma(\Delta)}{\Delta K (\Delta)^2}
=\beta^2  \frac{\sigma^2}{1-\beta^2\int_0^\infty  K ^2(s)\,ds}\frac{\Gamma(2\alpha-1)\Gamma(1-\alpha)}{\Gamma(\alpha)}.\end{equation}
Now, since $\Delta  K (\Delta)^2 \sim 1/\alpha^2 L(\Delta)^2
\Delta^{1-2\alpha}$ for $\Delta \to \infty$, \eqref{eq.exact_rate} follows immediately from
\eqref{eq.Ap_Kr} and \eqref{def.gamma}.
\end{proof}
\begin{example}
Let $k(t):=1/(1+t)^{1+\alpha}$ for  $1/2<\alpha<1.$ Then, \[
 K (x)=\frac{1}{\alpha}\frac{1}{(1+x)^\alpha} \quad
\mbox{and}\quad \int_0^\infty  K (x)^2\,
dx=\frac{1}{\alpha^2}\frac{1}{2\alpha-1}.\] If moreover
$\beta^2<\alpha^2(2\alpha-1)$ holds true, then assumptions of
Corollary \ref{th.concrete} are satisfied and we obtain
\[\gamma(\Delta)\sim \frac{\beta^2\sigma^2}{\alpha^2(2\alpha-1)-\beta^2}
\frac{\Gamma(2\alpha)\Gamma(1-\alpha)}{\Gamma(\alpha)}\Delta^{1-2\alpha}, \quad \Delta \to \infty. \]
\end{example}
We can also determine the rate of decay of $\gamma$ when $\alpha>1$.
\begin{corollary} \label{th.concreteL1}
Suppose that $k\in \mbox{RV}_\infty(-1-\alpha)$ with $\alpha>1$ and
$ K $ satisfies \eqref{eq.kappaintsmall}. Then, the function
$\gamma$ defined by \eqref{def.gamma} is integrable and satisfies
\begin{equation} \label{eq.exact_rateL1}
\gamma(\Delta)\sim\beta^2  \frac{\sigma^2}{1-\beta^2\int_0^\infty
 K ^2(s)\,ds} \int_0^\infty  K (s)\,ds \cdot
\frac{1}{\alpha}L(\Delta)\Delta^{-\alpha}, \quad \Delta \to \infty.
\end{equation}
\end{corollary}\begin{proof}
Since $\alpha>1$, it is clear that $\gamma$ is integrable. The proof
of the asymptotic behaviour of $\gamma$ follows from Theorem 5.2 in
\cite{Ap_Kr}.\end{proof}Whereas in the case $\alpha<1/2$ $ K $ does not satisfy $ K \in L^2(0, \infty)$,
the case $\alpha=1/2$ turns out to be critical: depending on the properties of the slowly varying
function $L$, both $ K \not \in L^2(0, \infty)$ as well
as $ K \in L^2(0, \infty)$ is possible.  In the latter
case, we can achieve arbitrary slow decay rates of the
autocovariance function.
\begin{corollary}\begin{enumerate}\item
Suppose that $k(t)=L(t)t^{-3/2}$, $t\ge
0$, with a slowly varying function $L$. Then, $ K \in L^2(0,
\infty)$ if and only if
\begin{equation}\label{eq.limit_finiteL}
\int_1^\infty \frac{L(t)^2}{t}\, dt<\infty.\end{equation} Moreover,
if \eqref{eq.limit_finiteL} holds true and $ K $ additionally satisfies \eqref{eq.kappaintsmall}, then
\[ \gamma(\Delta)\sim\frac{4\sigma^2}{1-\beta^2\int_0^\infty
 K ^2(s)\,ds} \int_\Delta^\infty \frac{L(s)^2}{s}\, ds, \quad \Delta\to \infty.\]
\item
Suppose that $f$ is in $C^1((0,\infty);(0,\infty))$,
$f(t)\to 0$ as $t\to\infty$ and that $-f'\in
\text{RV}_\infty(-1)$. Then $f\in \text{RV}_\infty(0)$ and
there exists $L\in \text{RV}_\infty(0)$ which satisfies
\eqref{eq.limit_finiteL} and
\begin{equation} \label{eq.constructLgammaasy}
\int_\Delta^\infty \frac{L(s)^2}{s}\,ds \sim f(\Delta), \quad\text{as
$\Delta\to\infty$}.
\end{equation}
\end{enumerate}
\end{corollary}
\begin{proof}
The proof follows from Theorem 3.6 and Corollary 3.7 in \cite{Ap_Kr}.
\end{proof}
\section{Connection with ARCH($\infty$) processes}
Let $(\Omega, \mathcal{F}, \mathbb{P})$ be a probability space
equipped with a filtration $\mathbb{F}=\{\mathcal{F}_{n}:n\in \mathbb{N}\}$.
Let us now consider the discrete version of equation
\eqref{eq.Xintro}. Let $\sigma>0$, $a>0$ and suppose $X=\{ X_n:n\in
\mathbb{N}\}$ satisfies
\[ X_{n+1}-X_n=\Bigl(\sigma+\beta\bigl( aX_n -\sum_{j=1}^n a_{n-j}X_j\bigr)\Bigr)\xi_{n+1}=:V_{n+1}\xi_{n+1}, \quad X_0=0,\]
where $(a_n)_{n\ge 0 }$ is a summable non--negative
sequence, $a=\sum_{j=0}^\infty a_j$ and $\xi=\{\xi_n \,:\, n\in
\mathbb{N}\}$ is a sequence of $\mathbb{F}$--adapted independent, identically distributed
random variables with $\mathbb{E}(\xi_n)=0, \,
\mathbb{E}(\xi_n^2)=1$ for all $n\in \mathbb{N}$. By  $U$ we denote
the uncorrelated process $U_n:=V_{n}\xi_{n}$, $n\ge 1$ and its
conditional variance process satisfies
\[\text{Var}(U_{n+1}| \mathcal{F}_n)=\mathbb{E}(U_{n+1}^2| \mathcal{F}_n)=V_{n+1}^2, \quad n\in \mathbb{N}.\]
Moreover,  using  $X_n=\sum_{j=1}^n X_j-X_{j-1}=\sum_{j=1}^n
V_j\xi_j$, we see that $V_n, \, n\in \mathbb{N}$ is given by
\begin{equation} \label{eq.Varchinfty}
V_{n+1} =\sigma+\beta \sum_{j=1}^n  \Bigl(a-\sum_{i=j}^n
a_{n-i}\Bigr) U_j =\sigma+ \beta\sum_{j=1}^n  K _{n-j} U_j =\sigma+
\beta\sum_{j=1}^n  K _{n-j} V_j\xi_j,
\end{equation}
where $( K _n)_{n\in \mathbb{N}}$ is a  non--negative
and
non--increasing sequence
whuch is
 given by $ K _n:=\sum_{j=n+1}^\infty
a_j$, $n\in \mathbb{N}$. Therefore
\begin{equation} \label{eq.Uarchinfty}
U_{n+1}=\Bigl(\sigma+\beta\sum_{j=1}^n  K _{n-j}
U_j\Bigr)\xi_{n+1}, \quad n\geq 1.
\end{equation}
It can be seen that this governing equation for $U$ is similar in
structure to that describing the dynamics of ARCH$(\infty)$
processes (see \cite{lgpkrl:2000}, \cite{pkrl:2000},
\cite{Zaffaroni:2004}). However, $U$ is not an ARCH($\infty$)
process because it is the squares of volatility that obeys an
autoregressive equation of the form \eqref{eq.Uarchinfty} in the
ARCH($\infty$) case, whereas here it is merely the process $U$ or
$V$ itself in \eqref{eq.Uarchinfty} or \eqref{eq.Varchinfty}.

For ARCH($\infty$) processes, under conditions that imply the weak
stationarity of $U^2$ (see \cite{Zaffaroni:2004}), the rate of decay
of the autocovariance function of $U^2$ can be determined (see
\cite{pkrl:2000}, \cite{lgpkrl:2000}, \cite{lgds:2002},
\cite{Zaffaroni:2004}). It has been shown that if $U^2$ is
stationary, then its autocovariance function must be summable.

However, by applying to equation \eqref{eq.Varchinfty} the methods
of this paper, it can be shown that the conditions that
$( K _n)_{n\in \mathbb{N}}$ is square--summable, and
$\beta^2\sum_{j=0}^\infty  K _j^2<1$ are equivalent to the
asymptotic stationarity of $V$. Moreover, the function $\Delta
\mapsto \lim_{n\to \infty} \text{Cov}(V_n, V_{n+\Delta})$ is not
summable over $\mathbb{N}$ if and only if $\sum_{j=0}^\infty
ja_j=\infty$.


\bibliographystyle{elsarticle-harv}
 
\end{document}